\documentclass[11pt]{amsart}

\usepackage{url}
\usepackage[T1]{fontenc}
\usepackage[dvips]{graphicx}

\usepackage[usenames]{color}

\usepackage{float}

\usepackage{amsmath}
\usepackage{amssymb}

\usepackage{arydshln}

\newtheorem{ThmCount}{Definition}[section]
\newtheorem{theorem}[ThmCount]{Theorem}

\newtheorem{lemma}[ThmCount]{Lemma} 

\newtheorem{note}{Remark}[section]

\def\G{\mathcal{G}}
\def\SG{\mathcal{S}}

\def\u{{\bf u}}

\def\S{{\bf S}}
\def\F{{\bf F}}
\def\x{{\bf x}}
\def\y{{\bf y}}

\def\X{{\bf X}}

\def\BSigma{{\bf \Sigma}}

\def\Phat{\hat{P}}
\def\phat{\hat{p}}

\newcommand{\rem}[1]{}

\begin{document}
\title{The Probabilistic Structure of Discrete Agent-Based Models}

\maketitle

\begin{abstract}
This paper describes a formalization of agent-based models (ABMs) as random walks on regular graphs and relates the symmetry group of those graphs to a coarse-graining of the ABM that is still Markovian.
An ABM in which $N$ agents can be in $\delta$ different states leads to a Markov chain with $\delta^N$ states.
In ABMs with a sequential update scheme by which one agent is chosen to update its state at a time, transitions are only allowed between system configurations that differ with respect to a single agent.
This characterizes ABMs as random walks on regular graphs.
The non-trivial automorphisms of those graphs make visible the dynamical symmetries that an ABM gives rise to because sets of micro configurations can be interchanged without changing the probability structure of the random walk.
This allows for a systematic loss-less reduction of the state space of the model.
\end{abstract}

\tableofcontents

\section{Introduction}

Agent-based models (ABMs) are an attempt to understand how macroscopic regularities may emerge through processes of self-organization in systems of interacting agents. 
The approach is first and foremost a computational methodology and the mathematical formalization of ABMs is still in its infancy.
This is probably due to the fact that a major motivation in the development of ABMs has been to relax a series of unrealistic assumptions made in other modeling frameworks just in order to keep mathematical tractability; namely, rationality, perfect information, agent homogeneity, among others.
The other side of the coin is that the focus on computer models and algorithms makes difficult the comparison of different models and also complicates a rigorous analysis of the model behavior.
In fact, the problems of code verification and model comparison including the discussion of standards for the replication of ABMs have nowadays become an area of research in its own (see e.g., \cite{Hales2003,Grimm2006,Galan2009}).
Many of those issues would probably be resolved with a sound mathematical formulation of an ABM.
On the other hand, it is also clear that the precise mathematical specification of a high-dimensional system of heterogeneous interacting agents along with their update mechanisms can be cumbersome. 

Agent-based systems are dynamical systems.
Typically implemented on a computer, the time evolution is computed as an iterative process -- an algorithm -- in which agents are updated according to the specified rules.
ABMs usually also involve a certain amount of stochasticity, because the agent choice and sometimes also the choice among different behavioral options is random.
This is why Markov chain theory is a good candidate for the mathematical formalization of ABMs.

To the authors knowledge, the first systematic approach to the development of mathematical formalism for ABMs in general is due to Laubenbacher and co-workers.
Ref. \cite{Laubenbacher2009} reviews existing formal frameworks that have the potential to model ABMs, such as cellular automata and finite dynamical systems and argue for the latter as an appropriate mathematical framework.
The possibility of using Markov chains in the analysis of ABMs has been pointed out in \cite{Izquierdo2009}.
The main idea is to consider all possible configurations of the agent system as the state space $\BSigma$ of a huge Markov chain.
While Ref. \cite{Izquierdo2009} mainly relies on numerical computations to estimate the stochastic transition matrices of the models, we have shown how to derive explicitly the transition probabilities $\Phat$ in terms of the update function $\u$ and a probability distribution $\omega$ accounting for the stochastic parts in the model (\cite{Banisch2012son,Banisch2013eccs}).
From general ABM to a particular class of models we refer to as single-step dynamics, this paper discusses in detail how to derive a microscopic Markov chain description (micro chain).
It turns out that ABMs with a sequential update scheme can be conceived as random walks on regular graphs.

This, in turn, hints at the possibility of reducing the state space of the microscopic Markov chain by exploiting systematically the dynamical symmetries that an ABM gives rise to.
Namely, the existence of non-trivial automorphisms of the micro chain tells us that certain sets of micro configurations can be interchanged without changing the probability structure of the random walk.
These sets of micro states can be aggregated or lumped into a single macro state and the resulting macro-level process is still a Markov chain.
In Markov chain theory, such a state space reduction by which no information about the dynamical behavior is lost is known as lumpability \cite{Burke1958,Rosenblatt1959,Kemeny1976,Rogers1981,Buchholz1994,Goernerup2008}.

Throughout the paper we  use the Voter Model (VM from now on) as a simple paradigmatic example (\cite{Kimura1964,Castellano2009}, among many others).
In the VM, agents can adopt two different states, which we may denote as white $\square$ and black $\blacksquare$.
The attribute could account for the opinion of an agent regarding a certain issue, its approval or disapproval regarding certain attitudes.
In an economic context $\blacksquare$ and $\square$ could encode two different behavioral strategies, or, in a biological context, the occurrence of mutants in a population of individuals.
The iteration process implemented by the VM is very simple.
At each time step, an agent $i$ is chosen at random along with one of its neighboring agents $j$ and one of them imitates the state of the other (by convention we assume the first to imitate the second).
In the long run, the model leads to a configuration in which all agents have adopted the same state (either $\square$ or $\blacksquare$).
In the context of biological evolution, this has been related to the fixation or extinction of a mutant in a population.
The VM has also been interpreted as a simplistic form of a social influence process by which a shared convention is established in the entire population.

This paper is organized as follows.
In Section 2, we discuss the general structure of ABMs in form of a graph of the possible model transitions.
Markovianity of the ABM process on the graph can be established by a so--called random mapping representation (Section 3).
In Section \ref{cha:3.SSD} a class of models giving rise to single-step dynamics and therefore to random walks on regular graphs is discussed.
Section 5 relates the symmetries in those graphs to partitions of the micro process with respect to which the chain is lumpable and Section 6 illustrates this at the example of a single-step model with $N$ agents that can be in three different states.
We summarize these results in Section 7.

\section{The Grammar of an ABM}

Let us consider an abstract ABM with finite configuration space $\BSigma = \S^N$ (meaning that there are $N$ agents with attributes $x_i \in \S$).
Any iteration of the model  (any run of the ABM algorithm) maps a configuration $\x \in \BSigma$ to another configuration $\y \in \BSigma$.
In general, the case that no agent changes such that $\x = \y$ is also possible.
Let us denote such a mapping by $F_z: \BSigma \rightarrow \BSigma$ and denote the set of all possible mappings by $\mathcal{F}$.
Notice that any element of $\mathcal{F}$ can be seen as a word of length $|\BSigma|$ over an $|\BSigma|$-ary alphabet, and there are $|\BSigma|^{|\BSigma|}$ such words \cite{Flajolet1990}:3.

Any $F_z \in \mathcal{F}$ induces a directed graph $(\BSigma,F_z)$ the nodes of which are the elements in $\BSigma$ (i.e., the agent configurations) and edges the set of ordered pairs $(\x,F_z(\x)), \forall \x\in\BSigma$.
Such a graph is called functional graph of $F_z$ because it displays the functional relations of the map $F_z$ on $\BSigma$.
That is, it represents the logical paths induced by $F_z$ on the space of configurations for any initial configuration $\x$.

Each iteration of an ABM can be thought of as a stochastic choice out of a set of deterministic options.
For an ABM in a certain configuration $\x$, there are usually several options (several $\y$) to which the algorithm may lead with a well-defined probability (see Fig. \ref{fig:AllAgentsChoicesVM}).
Therefore, in an ABM, the transitions between the different configurations $\x,\y,\ldots \in \BSigma$ are not defined by one single map $F_z$, but there is rather a subset $\mathcal{F}_Z \subset \mathcal{F}$ of maps out of which one map is chosen at each time step with a certain probability.
Let us assume we know all the mappings $\mathcal{F}_Z = \{F_1,\ldots,F_z,\ldots,F_n\}$ that are realized by the ABM of our interest.
With this, we are able to define a functional graph representation by $(\BSigma,\mathcal{F}_Z)$ which takes as the nodes all elements of $\BSigma$ (all agent configurations) and an arc $(\x,\y)$ exists if there is at least one $F_z \in \mathcal{F}_Z$ such that $F_z(\x) = \y$.
This graph defines the >>grammar<< of the system for it displays all the logically possible transitions between any pair of configurations of the model.

Consider the VM with three agents as an example.
In the VM agents have two possible states ($\S = \{\square,\blacksquare\}$) and the configuration space for a model of three agents is $\BSigma = \{\square,\blacksquare\}^3$.
In the iteration process, one agent $i$ is chosen at random along with one of its neighbors $j$ and agent $i$ imitates the state of $j$.
This means that $y_i = x_j$ after the interaction event.
Notice that once an agent pair $(i,j)$ is chosen the update is defined by a deterministic map $\u: \S^2 \rightarrow \S$.
Stochasticity enters first because of the random choice of $i$ and second through the random choice of one agent in the neighborhood.
Let us look at an example with three agents in the configuration $\x = (\square \blacksquare \blacksquare)$.
If the first agent is chosen ($i = 1$ and $x_1 = \square$) then this agent will certainly change state to $y_1 = \blacksquare$ because it will in any case meet a black agent.
For the second and the third agent ($i = 2$ or $i = 3$) the update result depends on whether one or the other neighbor is chosen because they are in different states.
Noteworthy, different agent choices may lead to the same configuration.
Here, this is the case if the agent pair $(2,3)$ or $(3,2)$ is chosen in which case the agent ($2$ or $3$) does not change its state because $x_2 = x_3$.
Therefore we have $\y = \x$ and there are two paths realizing that transition.

\begin{figure}[htbp]
	\centering
		\includegraphics[width=0.50\textwidth]{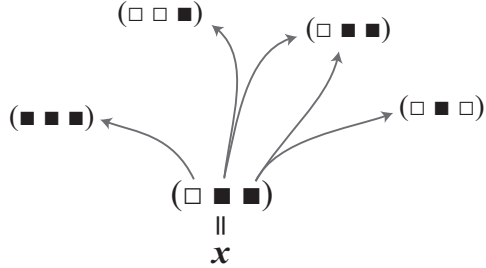}
	\caption{Possible paths from configuration $\x = (\square \blacksquare \blacksquare)$ in a small VM of three agents.}
	\label{fig:AllAgentsChoicesVM}
\end{figure}

In practice, the explicit construction of the entire functional graph may rapidly become a tedious task due to the huge dimension of the configuration space and the fact that one needs to check if $F_z(\x) = \y$ for each mapping $F_z \in \mathcal{F}_Z$ and all pairs of configurations $\x,\y$.
On the other hand, the main interest here is a theoretical one, because, as a matter of fact, a representation as a functional graph of the form $\Gamma = (\BSigma,\mathcal{F}_Z)$ exists for any model that comes in form of an iterated computer algorithm.
It is therefore a quite general way of formalizing ABMs and, as we will see in the sequel, it allows under some conditions to verify the Markovianity of the models at the micro level.

\section{From Functional Graphs to Markov Chains}

A functional graph $\Gamma = (\BSigma,\mathcal{F}_Z)$ defines the >>grammar<< of an ABM in the sense that it shows all possible transitions enabled by the model.
It is the first essential step in the construction of the Markov chain associated with the ABM at the micro level because there is a non-zero transition probability only if there is an arrow in the functional graph. 
Consequently, all that is missing for a Markov chain description is the computation of the respective transition probabilities.

For a class of models, including the VM, this is relatively simple because we can derive a random mapping representation \cite{Levin2009}:6/7 directly from the ABM rules.
Namely, if $F_{z_1}, F_{z_2},\ldots$ is a sequence of independent random maps, each having the same distribution $\omega$,  and $S_0 \in \BSigma$ has distribution $\mu_0$, then the sequence $S_0, S_1,\ldots$ defined by
\begin{equation}
S_t= F_{z_t}  (S_{t-1}), t \geq 1
\label{RMR1}
\end{equation}
is a Markov chain on $\BSigma$ with transition matrix $\Phat$:
\begin{equation}\label{RMR2}
\hat{P}(\x,\y)  = {\bf{Pr}_{\omega}}[z, F_z (\x) =  \y]; \x,\y \in \BSigma.
\end{equation}
Conversely \cite{Levin2009}, any Markov chain has a random map representation (RMR). 
Therefore, in that case, (\ref{RMR1}) and (\ref{RMR2}) may be taken as an equivalent definition of a Markov chain.
This is particularly useful in our case, because it shows that an ABM which can be described as above is, from a mathematical point of view, a Markov chain.
This includes several models described in \cite{Izquierdo2009}.

For general ABMs the explicit construction of a RMR can be cumbersome because it requires the dissection of stochastic and deterministic elements of the iteration procedure of the model.
In the VM, this separation is clear-cut and therefore a RMR is obtained easily.
In the VM, the random part consists of the choice of two connected agents $(i,j)$.
Once this choice is made we know that $y_i = x_j$ by the interaction rule.
This is sufficient to derive the >>grammar<< of the VM, because we need only to check one by one for all possible choices $(i,j)$ which transitions this choice induces on the configuration space.
For a system of three agents, with all agents connected to the other two, the set of functions $\mathcal{F}_Z = \{F_1,\ldots,F_z,\ldots,F_n\}$ is specified in Table \ref{tab:RMR_VM3}.
Notice that with three agents, there are 8 possible configurations indexed here by $a,b,\ldots,h$.
Moreover, there are 6 possible choices for $(i,j)$ such that $\mathcal{F}_Z$ consists of $n=6$ mappings.

\begin{table}[h]
\centering
\begin{tabular}{|c|c|c c c c c c c c|}
\hline 
$z$ 	& $(i,j)$ 	& a & b & c & d & e & f & g & h\\ 
	&  	& $\blacksquare\blacksquare\blacksquare$ & $\blacksquare\blacksquare\square$ & $\blacksquare\square\blacksquare$ & $\square\blacksquare\blacksquare$ & $\blacksquare\square\square$ & $\square\blacksquare\square$ & $\square\square\blacksquare$ & $\square\square\square$\\ 
\hline
$1$ 	& $(1,2)$ & a & b & g & a & h & b & g & h\\ 
$2$ 	& $(1,3)$ & a & f & c & a & h & f & c & h\\ 
$3$ 	& $(2,1)$ & a & b & a & g & b & h & g & h\\ 
$4$ 	& $(3,1)$ & a & a & c & f & c & f & h & h\\ 
$5$ 	& $(2,3)$ & a & e & a & d & e & h & d & h\\ 
$6$ 	& $(3,2)$ & a & a & e & d & e & d & h & h\\ 
\hline
\end{tabular} 
\caption{$\mathcal{F}_Z$ for the VM with three agents.}
\label{tab:RMR_VM3}
\end{table}

Each row of the table represents a mapping $F_z: \BSigma \rightarrow \BSigma$ by listing to which configuration $\y$ the respective map takes the configurations $a$ to $h$.
The first row, to make an example, represents the choice of the agent pair $(1,2)$.
The changes this choice induces depend on the actual agent configuration $\x$.
Namely, for any $\x$ with $x_1 = x_2$ we have $F_1(\x) = F_{(1,2)}(\x) = \x$.
So the configurations $a,b,g,h$ are not changed by $F_{(1,2)}$.
For the other configurations it is easy to see that $(\blacksquare\square\blacksquare)\rightarrow(\square\square\blacksquare)$ ($c \rightarrow g$), $(\square\blacksquare\blacksquare)\rightarrow(\blacksquare\blacksquare\blacksquare)$ ($d \rightarrow a$), $(\blacksquare\square\square)\rightarrow(\square\square\square)$ ($e \rightarrow h$), and $(\square\blacksquare\square)\rightarrow(\blacksquare\blacksquare\square)$ ($f \rightarrow b$).
Notice that the two configurations $(\square\square\square)$ and $(\blacksquare\blacksquare\blacksquare)$ with all agents equal are not changed by any map and correspond therefore to the final configurations of the VM.

In the RMR, we can use the possible agent choices $(i,j)$ in Table \ref{tab:RMR_VM3} directly to index the collection of maps $F_{(i,j)} \in \mathcal{F}_Z$.
We denote as $\omega(i,j)$ the probability of choosing the agent pair $(i,j)$ which corresponds to choosing the map $F_{(i,j)}$.
It is clear that we can proceed in this way in all models where the stochastic part concerns only the choice of agents.
Then, the distribution $\omega$ is independent of the current system configuration and the same for all times ($\omega(z_t) = \omega(z)$).
In this case, we obtain for the transition probabilities
\begin{equation}
\hat{P}(\x,\y)  = {\bf{Pr}_{\omega}}[(i,j), F_{(i,j)} (\x) =  \y] = \sum\limits_{\substack{(i,j):\\F_{(i,j)}(\x) = \y}}^{} \omega(i,j). 
\label{eq:PhatVM01}
\end{equation}
That is, the probability of transition from $\x$ to $\y$ is the conjoint probability $\sum\omega(i,j)$ of choosing an agent pair $(i,j)$ such that the corresponding map takes $\x$ to $\y$ (i.e., $F_{(i,j)} (\x) =  \y$).

\section{Single-Step Dynamics and Random Walks on Regular Graphs}
\label{cha:3.SSD}

In the sequel, we focus on a class of models which we refer to as \emph{single-step dynamics}.
They are characterized by the fact that only one agent changes at a time step.
Notice that this is very often the case in ABMs with a sequential update scheme and that sequential update is, as a matter of fact, the most typical iteration scheme in ABMs.
In terms of the >>grammar<< of these models, this means that non-zero transition probabilities are only possible between system configuration that differ in at most one position.
And this gives rise to random walks on regular graphs.

Consider a set of $N$ agents each one characterized by individual attributes $x_i$ that are taken in a finite list of possibilities $\S = \{1,\ldots,\delta\}$.
In this case, the space of possible agent configurations is $\BSigma = \S^N$.
Consider further a deterministic update function $\u: \S^r \times \Lambda \rightarrow \S$ which takes configuration $\x \in \BSigma$ at time $t$ to configuration $\y \in \BSigma$ at $t+1$ by
\begin{equation}
y_i = \u(x_i,x_j,\ldots,x_k,\lambda).
\label{eq:3.Update01}
\end{equation}
To go from one time step to the other in agent systems, usually, an agent $i$ is chosen first to perform a step.
The decision of $i$ then depends on its current state $(x_i)$ and the attributes of its neighbors $(x_j,\ldots,x_k)$.
The finite set $\Lambda$ accounts for a possible stochastic part in the update mechanism such that different behavioral options are implemented by different update functions $\u(\ldots,\lambda_1)$, $\u(\ldots,\lambda_2)$ etc.
Notice that for the case in which the attributes of the agents $(x_i,x_j,\ldots,x_k)$ uniquely determine the agent decision we have $\u: \S^r \rightarrow \S$ which strongly resembles the update rules implemented in cellular automata (CA).

As opposed to classical CA, however, a sequential update scheme is used in the class of models considered here.
In the iteration process, first, a random choice of agents along with a $\lambda$ to index the possible behavioral options is performed with probability $\omega(i,j,\ldots,k,\lambda)$.
This is followed by the application of the update function which leads to the new state of agent $i$ by Eq. (\ref{eq:3.Update01}).

Due to the sequential application of an update rule of the form $\u: \S^r \times \Lambda \rightarrow \S$ only one agent (namely agent $i$)  changes at a time so that all elements in $\x$ and $\y$ are equal except that element which corresponds to the agent that was updated during the step from $\x$ to $\y$.
Therefore, $x_j = y_j, \forall j \neq i$ and $x_i \neq y_i$.
We call $\x$ and $\y$ adjacent and denote this by $\x \stackrel{i}{\sim}  \y$.

It is then also clear that a transition from $\x$ to $\y$ is possible if $\x \stackrel{}{\sim} \y$.
Therefore, the adjacency relation $\stackrel{}{\sim}$ defines the >>grammar<< $\Gamma_{SSD}$ of the entire class of single-step models.
Namely, the existence of a map $F_z$ that takes $\x$ to $\y$, $\y = F_z(\x)$, implies that $\x \stackrel{i}{\sim}  \y$ for some $i \in \{1,\ldots,N\}$.
This means that any ABM that belongs to the class of single-step models performs a walk on $\Gamma_{SSD}$ or on a subgraph of it.

Let us briefly consider the structure of the graph $\Gamma_{SSD}$ associated to the entire class of single-step models.
From $\x \stackrel{i}{\sim} \y$ for $i = 1,\ldots,N$ we know that for any $\x$, there are $(\delta-1) N$ different vectors $\y$ which differ from $\x$ in a single position, where $\delta$ is the number of possible agent attributes.
Therefore, $\Gamma_{SSD}$ is a regular graph with degree $(\delta-1) N + 1$, because in our case, the system may loop by $y_i = x_i$.
As a matter of fact, our definition of adjacency as >>different in one position of the configuration<< is precisely the definition of so-called Hamming graphs which tells us that $\Gamma_{SSD} = H(N,\delta)$ (with loops).
In the case of the VM, where $\delta = 2$ we find $H(N,2)$ which corresponds to the $N$-dimensional hypercube.

As before, the transition probability matrix of the micro chain is denoted by $\Phat$ with $\Phat(\x,\y)$ being the probability for the transition from $\x$ to $\y$.
The previous considerations tell us that non-zero transition probabilities can exist only between two configurations that are linked in $H(N,\delta)$ plus the loop ($\Phat(\x,\x)$).
Therefore, each row of $\Phat$ contains no more than $(\delta-1) N +1$ non-zero entries.
In the computation of $\Phat$ we concentrate on pairs of adjacent configurations.
For $\x \stackrel{i}{\sim}  \y$ with $x_i \neq y_i$ we have
\begin{equation} 
\Phat(\x,\y) = \sum\limits_{\substack{(i,j,\ldots,k,\lambda): \\y_i = \u(x_i,x_j,\ldots,x_k,\lambda)} }^{} \omega(i,j,\ldots,k,\lambda)
\label{eq:PhatSSD}
\end{equation}
which is the conjoint probability to choose agents and a rule $(i,j,\ldots,k,\lambda)$ such that the $i$th agent changes its attribute by $y_i = \u(x_i,x_j,\ldots,x_k,\lambda)$.
For the probability that the model remains in $\x$, $\Phat(\x,\x)$, we have
\begin{equation}
\Phat(\x,\x) = 1 - \sum^{}_{\substack{\y {\sim} \x}} \Phat(\x,\y).
\label{eq:PhatSSD02}
\end{equation}
Eq. (\ref{eq:PhatSSD}) makes visible that the probability distribution $\omega$ plays the crucial role in the computation of the elements of $\Phat$, a fact that has been exploited in Ref. \cite{Banisch2013acs}.

\section{Graph Symmetries and Markov Chain Aggregation}

Markov chain aggregation concerns the question of what happens when the micro-level process -- defined by the micro chain $(\BSigma,\Phat)$ --  is projected onto a coarser partition of the state space $\BSigma$.
Such a situation naturally arises if the ABM is observed not at the micro level of $\BSigma$, but rather in terms of a measure $\phi$ on $\BSigma$ by which all configuration in $\BSigma$ that give rise to the same measurement are mapped into the same macro state, say $X_k \in \X$.
The first important question then concerns the \emph{lumpability} of the micro chain with respect to the partition $\X$.
In the case of lumpability, the resulting macro-level process is still a Markov chain and the transition rates $P$ can be obtained in a relatively simple way from the microscopic transition matrix $\Phat$.
Fig. \ref{fig:ProjectionGeneral} illustrates such a projection construction.

\begin{figure}[hbtp]
\centering
\includegraphics[width=0.95\linewidth]{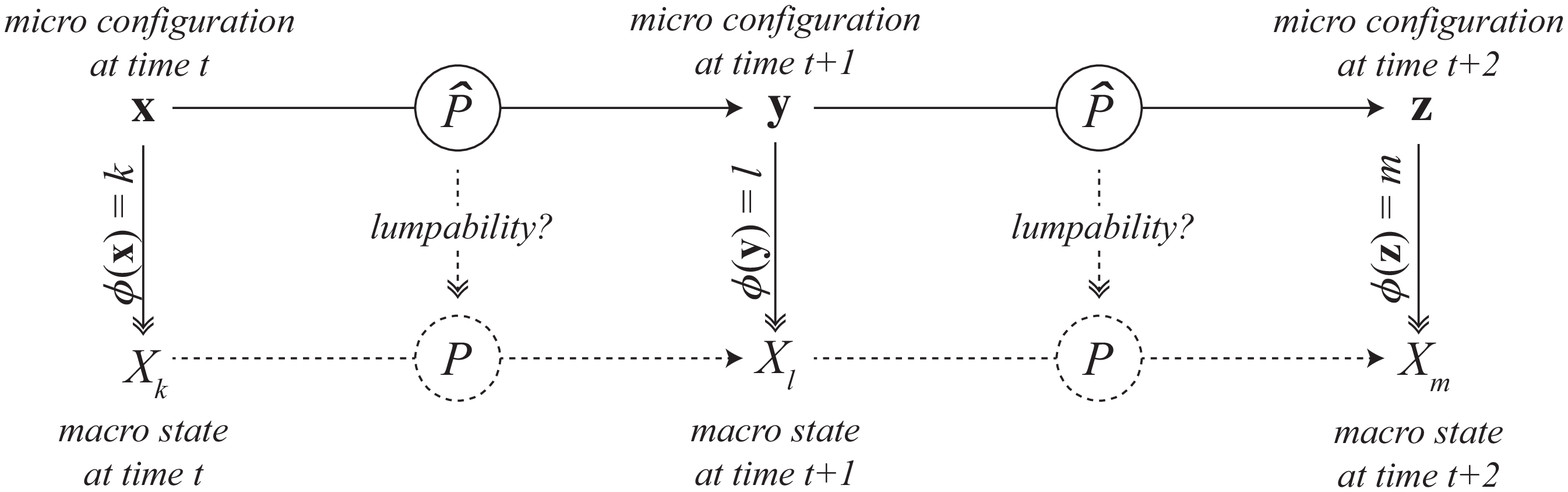}
\caption{A micro process ($\x,\y, {\bf z} \in \BSigma$) is observed ($\phi$) at a higher level and this observation defines another macro level process ($X_k,X_l,X_m \in \X$). The micro process is a Markov chain with transition matrix $\Phat$. The macro process is a Markov chain (with $P$) only in the case of lumpability.}
\label{fig:ProjectionGeneral}
\end{figure}

Necessary and sufficient conditions for lumpability are provided by Thm. 6.3.2 in \cite{Kemeny1976}.
Let us denote by $\phat_{\x Y} $ the conjoint probability for $\x$ to go to elements $\y \in Y$ where $Y \subseteq \BSigma$ is a subset of the configuration space.
Thm. 6.3.2 in \cite{Kemeny1976} states that a Markov chain $(\Phat,\BSigma)$ is \emph{lumpable} with respect to a partition $\X = (X_1,\ldots,X_P)$ if for every two subsets $X_k$ and $X_l$ the sum $\hat{p}_{\x X_l} = \sum\limits_{\y \in X_l}^{} \hat{p}_{\x \y}$ is equal for all $\x \in X_k$.
Moreover, these common values form the transition probabilities $P(X_k,X_l)$ for a new chain $(P,\X)$.

Since ABM micro chains can be seen as random walks on regular graphs, it is convenient to provide lumpability conditions in form of the symmetry structure of the micro chain.
Let us restate the respective result (previously introduced in a similar form in \cite{Banisch2012son}, Prop. 3.1):
\begin{theorem}
\label{thm:symmetry}
Let  $(\bf{\Sigma}, \hat{P})$ be a Markov chain and $\x,\y$ elements of $\BSigma$.
Consider a partition $\X$ of $\BSigma$ and along with $\X$ a transformation group $\G$ acting on $\BSigma$ that generates $\X$.
(That is, the orbits of $\G$ on $\BSigma$ construct $\X$.)
If the Markov transition probability $\hat{P}$ is symmetric with respect to $\G$,
\begin{equation}\label{eq:symmetry_lumpability}
 {\hat P} (\x,\y) = \: {\hat P} ({\hat{\sigma}}(\x),{\hat{\sigma}}(\y)) : \forall  {\hat{\sigma}} \in {\G},
\end{equation}
the partition $\X = (X_1,\dots, X_{n})$ is (strongly) lumpable.
\end{theorem}

\begin{proof}
For the proof it is sufficient to show that any two configurations $\x$ and $\x'$ with $\x' = \hat{\sigma}(\x)$ satisfy 
\begin{equation}
\phat_{\x Y} = \sum\limits_{\y \in Y}^{} \hat{P}(\x,\y) = \sum\limits_{\y \in Y}^{} \hat{P}(\x',\y) = \phat_{\x' Y}
 \label{eq:3.macroEquivalence01}
\end{equation}
for all $Y \in \X$. 
Consider any two subsets $X,Y \in \X$ and take $\x \in X$.
Because $\G$ preserves the partition it is true that $\x' \in X$.
Now we have to show that Eq. (\ref{eq:3.macroEquivalence01}) holds.
First the probability for $\x' = \hat{\sigma}(\x)$ to go to an element $\y \in Y$ is
\begin{equation}
\phat_{\hat{\sigma}(\x) Y} = \sum\limits_{\y \in Y}^{} \hat{P}(\hat{\sigma}(\x),\y).
\end{equation}
Because the $\hat{\sigma}$ are bijections and preserve $\X$ we have $\hat{\sigma}(Y) = Y$ and there is for every $\y \in Y$ exactly one $\hat{\sigma}(\y) \in Y$. Therefore we can substitute
\begin{equation}
\phat_{\hat{\sigma}(\x) Y} = \sum\limits_{\y \in Y}^{} \hat{P}(\hat{\sigma}(\x),\hat{\sigma}(\y)) = \sum\limits_{\y \in Y}^{} \hat{P}(\x,\y) = \phat_{\x Y},
\end{equation}
where the second equality comes by the symmetry condition (\ref{eq:symmetry_lumpability}) that $\Phat(\x,\y) = \Phat(\hat{\sigma}(\x),\hat{\sigma}(\y))$.
\end{proof}

The usefulness of the lumpability conditions stated in Thm. \ref{thm:symmetry} becomes apparent recalling that ABMs can be seen as random walks on regular graphs defined by the functional graph or >>grammar<< of the model $\Gamma = (\BSigma,\mathcal{F}_Z)$.
The full specification of the micro process $(\BSigma, \hat{P})$ is obtained by assigning transition probabilities to the connections in $\Gamma$ and we can interpret this as a weighted graph.
The regularities of $(\BSigma, \hat{P})$ are captured by a number of non-trivial automorphisms which, in the case of ABMs, reflect the symmetries of the models.

In fact, Thm. \ref{thm:symmetry} allows to systematically exploit the symmetries of an agent model in the construction of partitions with respect to which the micro chain is lumpable.
Namely, the symmetry requirement in Thm. \ref{thm:symmetry}, that is, Eq. (\ref{eq:symmetry_lumpability}), corresponds precisely to the usual definition of automorphisms of $(\BSigma, \hat{P})$.
The set of all permutations $\hat{\sigma}$ that satisfy (\ref{eq:symmetry_lumpability}) corresponds then to the automorphism group of $(\BSigma, \hat{P})$.
\begin{lemma}
Let $\G$ be the automorphism group of the micro chain $(\BSigma, \hat{P})$.
The orbits of $\G$ define a lumpable partition $\X$ such that every pair of micro configurations $\x,\x' \in \BSigma$ for which $\exists\hat{\sigma} \in \G$ such that $\x' = \hat{\sigma}(\x)$ belong to the same subset $X_i \in \X$.
\label{thm:Lambda}
\end{lemma}

\begin{note}
Lemma \ref{thm:Lambda} actually applies to any $\G$ that is a proper subgroup of the automorphism group of $(\BSigma, \hat{P})$.
The basic requirement for such a subset $\G$ to be a group is that be closed under the group operation which establishes that $\hat{\sigma}(X_i) = X_i$.
With the closure property, it is easy that any such subgroup $\G$ defines a lumpable partition in the sense of Thm. \ref{thm:symmetry}.
\end{note}

\section{Groups of Automorphisms, Macro Chains and System Properties}

In this section we illustrate the previous ideas at the example of three state single-step dynamics.
Consider a system of $N$ agents each one characterized by an attribute $x_i \in \{a,b,c\}$, that is $\delta = 3$.
As discussed in Section \ref{cha:3.SSD}, the corresponding graph $\Gamma$ encoding all the possible transitions is the Hamming graph $H(N,3)$.
The nodes $\x,\y$ in $H(N,3)$ correspond to all possible agent combinations and are written as vectors $\x = (x_1,\ldots,x_N)$ with symbols $x_i \in \{a,b,c\}$.
The automorphism group of $H(N,3)$ is composed of two groups generated by operations changing the order of elements in the vector (agent permutations) and by permutations acting on the set of symbols $\S=\{a,b,c\}$ (agent attributes).
Namely, it is given by the direct product 
\begin{equation}
Aut(H(N,\delta)) = \SG_N \otimes \SG_{\delta}
\label{eq:AutHNdelta}
\end{equation}
of the symmetric group $\SG_N$ acting on the agents and the 
group $\SG_{\delta}$ acting on the agent attributes.

Let us first look at a very small system of $N = 2$ agents and $\delta = 3$ states.
The corresponding microscopic structure -- the graph $H(2,3)$ -- is shown on the l.h.s. of Fig. \ref{fig:MicroVM.2agents3states.AutW}.
It also illustrates the action of $\SG_N$ on the $\x,\y \in \BSigma$, that is, the bijection induced on the configuration space by permuting the agent labels.
Noteworthy, in the case of $N=2$ there is only one alternative ordering of agents denoted here as $\hat{\sigma}_{\omega}(\x)$ which takes $(x_1,x_2) 
\stackrel{\hat{\sigma}_{\omega}}{\longleftrightarrow} (x_2,x_1)$.
The respective group $\SG_{N=2}$ therefore induces a partition in which all configurations $\x,\y$ with the same number of attributes $a,b,c$ are \emph{lumped} into the same set, which we may denote as $X_{\langle k_a,k_{b},k_c\rangle}$.
See r.h.s. of Fig. \ref{fig:MicroVM.2agents3states.AutW}.


\begin{figure}[htbp]
	\centering
		\includegraphics[width=0.80\textwidth]{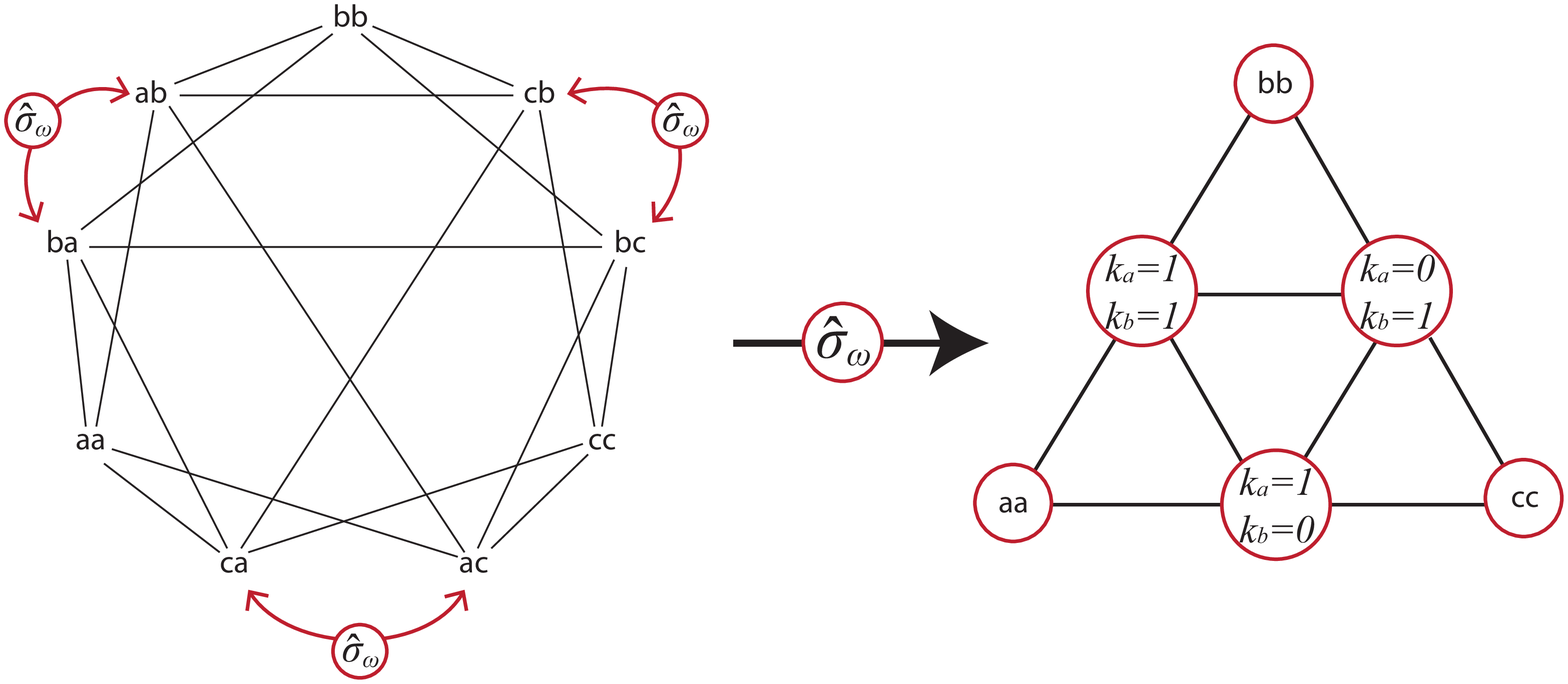}
	\caption{$H(2,3)$ and the reduction induced by $\SG_N$.}
	\label{fig:MicroVM.2agents3states.AutW}
\end{figure}

\begin{figure}[ht]
	\centering
		 \includegraphics[width=.90\textwidth]{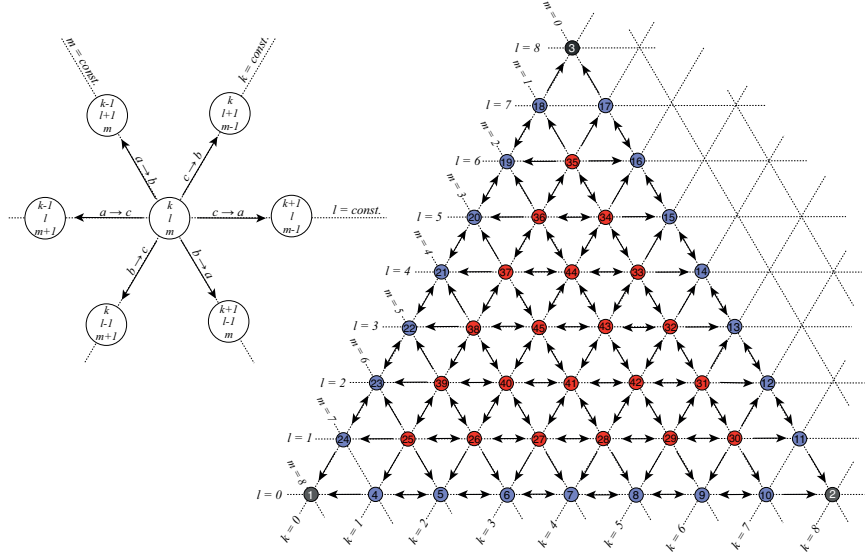}
	\caption{Transition structure (l.h.s) and state topology (r.h.s) of the VM with three attributes for $N=8$.}
	\label{fig:StateTopologyAndTransitions.T3.N8}
\end{figure}

More generally in the case of $N$ agents and $\delta$ agent attributes the group $\SG_{N}$ induces a partition of the configuration space $\BSigma$ by which all configurations with the same attribute frequencies are collected in the same macro set.
Let us define $N_s (\x)$ to be the number of agents in the configuration $\x$ with attribute $s$, $s = 1,\ldots,\delta$,  and then $X_{\langle k_1, k_2, \dots, k_{\delta}\rangle} \subset \BSigma$ as
\begin{equation}
\begin{split}
X_{\langle k_1,  \dots, k_s, \dots, k_{\delta} \rangle}  = \left\{ \vphantom{\sum_{s=1}^{\delta}}\x \in \BSigma \ : N_1(\x) = k_1, \dots, N_s(\x) = k_s, \dots\right.\\\left.\ldots, N_{\delta} (\x) = k_{\delta} \mbox{ and } \ \sum_{s=1}^{\delta}  k_{s} = N\right\}.
\end{split}
\label{eq:X< >}
\end{equation}
Each $X_{\langle k_1, k_2, \dots, k_{\delta}\rangle}$ contains all the configurations $\x$ in which exactly $k_s$ agents hold attribute $s$ for any $s$.
We use the notation $\langle k_1, k_2, \dots, k_{\delta }\rangle$ to indicate that $\sum_{s=1}^{\delta }  k_{s} = N$.
Therefore, the reduced state space is organized as a $\delta$-simplex lattice, see Fig. \ref{fig:StateTopologyAndTransitions.T3.N8}.

For a model with $N=8$ and $\delta = 3$ the resulting reduced state space is shown in Fig. \ref{fig:StateTopologyAndTransitions.T3.N8}.
The transition structure depicted in Fig. \ref{fig:StateTopologyAndTransitions.T3.N8} corresponds to the VM.
The number of $a$, $b$ and $c$ agents is denoted by (respectively) $k$, $l$ and $m$ so that $\X =\{ X_{\langle k,l,m \rangle} : 0 \leq k,l,m \leq N, k+l+m = N \}$.
The number of states for a system with $N$ agents is $S = \sum_{i = 0}^N (i+1) = \frac{(N + 1)(N + 2)}{2}$.

For Voter-like models -- used, for instance, as models of opinion and social dynamics -- it is not unusual to study the dynamical behavior by looking at the time evolution of the respective attribute frequencies.
It is important to notice, however, that the resulting partition is lumpable only if the transition matrix $\Phat$ is symmetric with respect to the action of $\SG_N$ on $\BSigma$, namely if Thm. \ref{thm:symmetry} holds for $\SG_N$.
We have shown in \cite{Banisch2012son} that this is only true for homogeneous mixing and the case of inhomogeneous interaction topologies is discussed in~\cite{Banisch2013acs}.

\begin{figure}[htbp]
	\centering
		\includegraphics[width=0.95\textwidth]{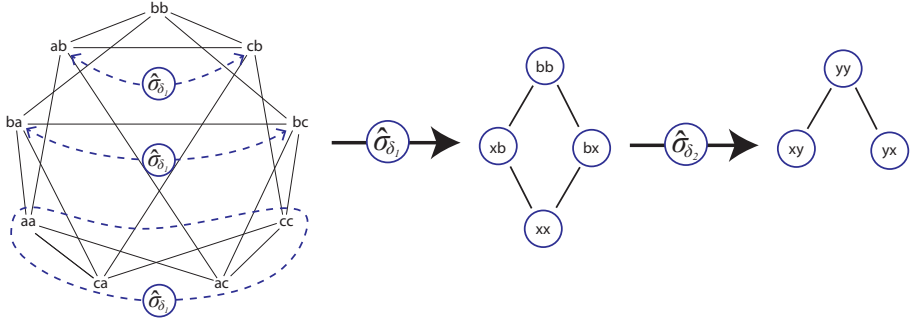}
	\caption{$H(2,3)$ and the reductions induced by $\SG_{\delta}$.}
	\label{fig:MicroVM.2agents3states.AutD}
\end{figure}

Let us now consider $\SG_{\delta}$.
On the l.h.s. of Fig. \ref{fig:MicroVM.2agents3states.AutD} the graph $H(2,3)$ is shown along with the bijections on it induced by permutation of attributes $a$ and $c$, $abc \stackrel{\hat{\sigma}_{\delta_1}}{\longleftrightarrow} cba )$.
Effectively, this corresponds to the situation of looking at >>one attribute ($b$) against the other two ($x = a \cup c$)<<.
Noteworthy, taking that perspective (see graph in the middle of Fig. \ref{fig:MicroVM.2agents3states.AutD}) corresponds to a reduction of $H(2,3)$ to $H(2,2)$ or, more generally, of $H(N,3)$ to the hypercube $H(N,2)$.
This means that, under the assumption of symmetric agent rules with respect to the attributes, single-step models with $\delta$ states are reducible to the binary case.

Moreover, even the binary case allows for further reduction (see r.h.s. of Fig. \ref{fig:MicroVM.2agents3states.AutD}).
Namely, assuming the additional symmetry $bx \stackrel{\hat{\sigma}_{\delta_2}}{\longleftrightarrow} xb )$ corresponding in a binary setting to the simultaneous flip of all agent states $x_i \rightarrow \bar{x}_i, \forall i$.
The VM is a nice example in which independent of the interaction topology, $\Phat(\x,\y) = \Phat(\bar{\x},\bar{\y})$.
This reduces the state space to one half of $H(N,2)$, which we shall denote as $H_{1/2}(N,2)$.

\begin{figure}[htbp]
	\centering
		\includegraphics[width=0.95\textwidth]{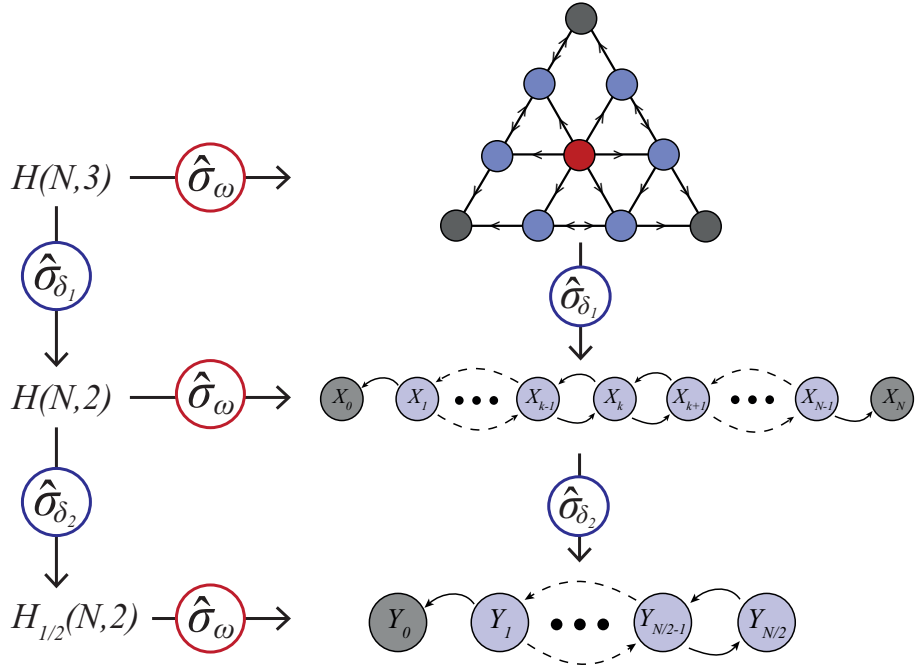}
	\caption{Different levels of description are associated to different symmetry groups of $H(N,3)$.}
	\label{fig:AutProjectionScheme}
\end{figure}

The most interesting reductions can be reached by the combination of $\SG_N$ and $\SG_{\delta}$.
Fig. \ref{fig:AutProjectionScheme} shows possible combinations and the resulting macroscopic state spaces starting from $H(N,3)$.
For instance, partitioning $H(N,3)$ by using the set of agent permutations $\SG_N$ leads to state space organized as a triangular lattice (see also Fig. \ref{fig:StateTopologyAndTransitions.T3.N8}).
Lumpability of the micro process $(\BSigma,\Phat)$ on $H(N,3)$ with respect to this state space rests upon the symmetry of the agent interaction probabilities with respect to all agent permutations (\cite{Banisch2013acs}, see \cite{Banisch2014acscodym} for a discussion of the non-lumpable case).
From the triangular structure shown on the upper right in Fig. \ref{fig:AutProjectionScheme}, a further reduction ca be obtained by taking into account the symmetry of the interaction rules with respect to (at least) one pair of attributes, which we have denoted as $\hat{\sigma}_{\delta_1}$.
The resulting macro process on $\X = (X_0,\ldots,X_N)$ is a random walk on the line with $N+1$ states, known as Moran process for the VM interaction (after \cite{Moran1958}).
In a binary setting, the macro states $X_k$ collect all micro configurations with $k$ agents in state $\square$ (and therefore $N-k$ agents in $\blacksquare$).
Notice that a Markov projection to the Moran process is possible also for $\delta > 3$ if the micro process is symmetric with respect to permutations of (at least) $\delta-1$ attributes.
The group of transformations associated to this partition may be written as $\SG_N \otimes \SG_{\delta-1} \subset Aut(H(N,\delta))$.

The reduction obtained by using the full automorphism group of $H(N,3)$ is shown on the bottom of Fig. \ref{fig:AutProjectionScheme}.
With respect to the Moran process on $\X = (X_0,\ldots,X_N)$, it means that the pairs $\{X_k,X_{(N-k)}\}$ are lumped into the same state $Y_k$.
This can be done if we have for any $k$, $P(X_k,X_{k \pm 1}) = P(X_{(N-k)},X_{(N-k) \mp 1})$.
As a matter of fact, this description still captures the number of agents in the same state, but now information about in which state they are is omitted.
This is only possible (lumpable) if the model implements completely symmetric interaction rules.

\section{Summary}

This paper analyses the probabilistic structure of a class of agent-based models (ABMs).
In an ABM in which $N$ agents can be in $\delta$ different states there are $\delta^N$ possible agent configurations and each iteration of the model takes one configuration into another one.
It is therefore convenient to conceive of the agent configurations as the nodes of a huge directed graph and to link two configurations $\x,\y$ whenever the application of the ABM rules to $\x$ may lead to $\y$ in one step.
If a model operates with a sequential update scheme by which one agent is chosen to update its state at a time, transitions are only allowed between system configurations that differ with respect to a single element (agent).
The graph associated to those single-step models is the Hamming graph $H(N,\delta)$.

The fact that a single-step ABM corresponds to a random walk on a regular graph allows for a systematic study of the symmetries in the dynamical structure of an ABM.
Namely, the existence of non-trivial automorphisms of the ABM micro chain tells us that certain sets of agent configurations can be interchanged without changing the probability structure of the random walk.
These sets of micro states can be aggregated or lumped into a single macro state and the resulting macro-level process is still a Markov chain.
If the microscopic rules are symmetric with respect agent ($\SG_N$) and attribute ($\SG_{\delta}$) permutations the full automorphism group of $H(N,\delta)$ is realized and allows for a reduction from $\delta^N$ micro to around $N/2$ macro states.
Moreover, different combinations of subgroups of automorphisms and the reductions they imply are rather meaningful in terms of observables and system properties.

Notice finally that other update schemes (beyond single-step dynamics) -- even the case of synchronous update\footnote{The author thanks J\"{u}rgen Jost for this suggestion.} -- do not necessarily affect the symmetries of the micro process.
The described approach may be applied to these cases as well.
Extending the framework to models with continuous agent attributes is another challenging issue to be addressed by future work.


\begin{thebibliography}{10}

\bibitem{Banisch2014acscodym}
S.~Banisch.
\newblock {From Microscopic Heterogeneity to Macroscopic Complexity in the
  Contrarian Voter Model}, 2014.
\newblock Submitted. 

\bibitem{Banisch2013acs}
S.~Banisch and R.~Lima.
\newblock {Markov Chain Aggregation for Simple Agent-Based Models on Symmetric
  Networks: The Voter Model}, 2013.
\newblock Submitted. 

\bibitem{Banisch2012son}
S.~Banisch, R.~Lima, and T.~Ara\'{u}jo.
\newblock {Agent Based Models and Opinion Dynamics as Markov Chains}.
\newblock {\em Social Networks}, 34:549--561, 2012.

\bibitem{Banisch2013eccs}
S.~Banisch, R.~Lima, and T.~Ara{\'u}jo.
\newblock {Aggregation and Emergence in Agent-Based Models: A Markov Chain
  Approach}.
\newblock In T.~Gilbert, M.~Kirkilionis, and G.~Nicolis, editors, {\em
  Proceedings of the European Conference on Complex Systems 2012}, Springer
  Proceedings in Complexity, pages 3--7. Springer International Publishing,
  2013.

\bibitem{Buchholz1994}
P.~Buchholz.
\newblock {Exact and Ordinary Lumpability in Finite {Markov} Chains}.
\newblock {\em J. Appl. Prob.}, 31(1):59--75, 1994.

\bibitem{Burke1958}
C.~J. Burke and M.~Rosenblatt.
\newblock {A Markovian Function of a Markov Chain}.
\newblock {\em The Annals of Mathematical Statistics}, 29(4):1112 -- 1122,
  1958.

\bibitem{Castellano2009}
C.~Castellano, S.~Fortunato, and V.~Loreto.
\newblock Statistical physics of social dynamics.
\newblock {\em Reviews of Modern Physics}, 81(2):591--646, 2009.

\bibitem{Flajolet1990}
P.~Flajolet and A.~M. Odlyzko.
\newblock {Random Mapping Statistics}.
\newblock In {\em Advances in Cryptology}, pages 329--354. Springer Verlag,
  1990.

\bibitem{Galan2009}
J.~M. Gal\'{a}n, L.~R. Izquierdo, S.~S. Izquierdo, J.~I. Santos, R.~del Olmo,
  A.~L\'{o}pez-Paredes, and B.~Edmonds.
\newblock {Errors and Artefacts in Agent-Based Modelling}.
\newblock {\em Journal of Artificial Societies and Social Simulation}, 12(1):1,
  2009.

\bibitem{Goernerup2008}
O.~G\"{o}rnerup and M.~N. Jacobi.
\newblock {A Method for Inferring Hierarchical Dynamics in Stochastic
  Processes}.
\newblock {\em Advances in Complex Systems}, 11(1):1--16, 2008.

\bibitem{Grimm2006}
V.~Grimm, U.~Berger, F.~Bastiansen, S.~Eliassen, V.~Ginot, J.~Giske,
  J.~Goss-Custard, T.~Grand, S.~K. Heinz, G.~Huse, A.~Huth, J.~U. Jepsen,
  C.~Jorgensen, W.~M. Mooij, B.~Muller, G.~Pe'er, C.~Piou, S.~F. Railsback,
  A.~M. Robbins, M.~M. Robbins, E.~Rossmanith, N.~Ruger, E.~Strand, S.~Souissi,
  R.~A. Stillman, R.~Vabo, U.~Visser, and D.~L. Deangelis.
\newblock A standard protocol for describing individual-based and agent-based
  models.
\newblock {\em Ecological Modelling}, 198:115--126, 2006.

\bibitem{Hales2003}
D.~Hales, J.~Rouchier, and B.~Edmonds.
\newblock {Model--to--Model Analysis}.
\newblock {\em Journal of Artificial Societies and Social Simulation}, 6(4),
  2003.

\bibitem{Izquierdo2009}
L.~R. Izquierdo, S.~S. Izquierdo, J.~M. Gal\'{a}n, and J.~I. Santos.
\newblock {Techniques to Understand Computer Simulations: Markov Chain
  Analysis}.
\newblock {\em Journal of Artificial Societies and Social Simulation}, 12(1):6,
  2009.

\bibitem{Kemeny1976}
J.~G. Kemeny and J.~L. Snell.
\newblock {\em {Finite Markov Chains}}.
\newblock Springer, 1976.

\bibitem{Kimura1964}
M.~Kimura and G.~H. Weiss.
\newblock The stepping stone model of population structure and the decrease of
  genetic correlation with distance.
\newblock {\em Genetics}, 49:561--576, 1964.

\bibitem{Laubenbacher2009}
R.~C. Laubenbacher, A.~S. Jarrah, H.~S. Mortveit, and S.~S. Ravi.
\newblock {Agent Based Modeling, Mathematical Formalism for}.
\newblock In {\em Encyclopedia of Complexity and Systems Science}, pages
  160--176. 2009.

\bibitem{Levin2009}
D.~A. Levin, Y.~Peres, and E.~L. Wilmer.
\newblock {\em {Markov chains and mixing times}}.
\newblock American Mathematical Society, 2009.

\bibitem{Moran1958}
P.~A.~P. Moran.
\newblock Random processes in genetics.
\newblock In {\em Proceedings of the Cambridge Philosophical Society},
  volume~54, pages 60--71, 1958.

\bibitem{Rogers1981}
L.~C.~G. Rogers and J.~W. Pitman.
\newblock {Markov Functions}.
\newblock {\em The Annals of Probability}, 9(4):573--582, 1981.

\bibitem{Rosenblatt1959}
M.~Rosenblatt.
\newblock {Functions of a Markov Process that are Markovian}.
\newblock {\em Journal of Mathematics and Mechanics}, 8(4):134 --145, 1959.

\end{thebibliography}

\end{document}